\documentclass[english]{article}

\usepackage[T1]{fontenc}
\usepackage[letterpaper]{geometry}
\usepackage{xcolor}
\usepackage{array}
\usepackage{float}
\usepackage{multirow}
\usepackage{amsmath}
\usepackage{amsthm}
\usepackage{undertilde}
\usepackage{graphicx}
\usepackage{setspace}
\PassOptionsToPackage{normalem}{ulem}
\usepackage{ulem}
\makeatletter

\providecommand{\tabularnewline}{\\}
\floatstyle{ruled}
\newfloat{algorithm}{tbp}{loa}
\providecommand{\algorithmname}{Algorithm}
\floatname{algorithm}{\protect\algorithmname}
\providecolor{lyxadded}{rgb}{0,0,1}
\providecolor{lyxdeleted}{rgb}{1,0,0}

\DeclareRobustCommand{\lyxsout}[1]{\ifx\\#1\else\sout{#1}\fi}

\numberwithin{equation}{section}
\numberwithin{figure}{section}
\theoremstyle{plain}
\newtheorem{thm}{\protect\theoremname}
\theoremstyle{plain}
\newtheorem{prop}[thm]{\protect\propositionname}
\theoremstyle{plain}
\newtheorem{fact}[thm]{\protect\factname}
\theoremstyle{definition}
\newtheorem{example}[thm]{\protect\examplename}
\theoremstyle{plain}
\newtheorem{lem}[thm]{\protect\lemmaname}

\usepackage{algpseudocode}
\allowdisplaybreaks
\abovedisplayshortskip
\belowdisplayshortskip
\date{}
\providecommand{\keywords}[1]{\textbf{\textit{Keywords---}} #1}

\makeatother

\usepackage{babel}
\providecommand{\examplename}{Example}
\providecommand{\factname}{Fact}
\providecommand{\lemmaname}{Lemma}
\providecommand{\propositionname}{Proposition}
\providecommand{\theoremname}{Theorem}

\begin{document}
\title{On Competitive Analysis for Polling Systems}
\author{Jin Xu\thanks{Corresponding author, Email: jinxu@tamu.edu, Industrial and Systems
Engineering, Texas A\&M University, TX 77843, USA}\and Natarajan Gautam\thanks{Email: gautam@tamu.edu, Industrial and Systems Engineering, Texas
A\&M University, TX 77843, USA}}
\maketitle
\begin{abstract}
Polling systems have been widely studied, however most of these studies
focus on polling systems with renewal processes for arrivals and random
variables for service times. There is a need driven by practical applications
to study polling systems with arbitrary arrivals (not restricted to
time-varying or in batches) and revealed service time upon a job's
arrival. To address that need, our work considers a polling system
with generic setting and for the first time provides the worst-case
analysis for online scheduling policies in this system. We provide
conditions for the existence of constant competitive ratios, and competitive
lower bounds for general scheduling policies in polling systems. Our
work also bridges the queueing and scheduling communities by proving
the competitive ratios for several well-studied policies in the queueing
literature, such as cyclic policies with exhaustive, gated or $l$-limited
service disciplines for polling systems. 
\end{abstract}

\section{Introduction}

This study has been motivated by operations in smart manufacturing
systems. As an illustration, consider a 3D printing machine that uses
a particular material informally called ``ink'' to print. Jobs of
the same prototype are printed using the same ink, and when a different
prototype (for simplicity, say a different color) is to be printed,
a different ink is required and the machine undergoes a setup that
takes time to switch inks. The unprocessed jobs of the same prototype
can be regarded as a ``queue''. This problem can thus be modeled
as a polling system where the server polls a queue, processes jobs,
incurs a setup time, processes another queue, and so on. In practice,
besides the ink (material), other factors such as processing temperature,
equipment setting and other processing requirements that are required
by different job prototypes will also incur setup times. 

Another interesting feature of 3D printing is that it is possible
to reveal the processing time of each job upon the job's arrival.
This is because, before getting printed, the printing requirements
such as temperature, nozzle route, printing speed, and so on are specified
for the job, and using that we can easily acquire the printing time
before processing. Therefore, it is unnecessary to assume that the
processing time of a job is stochastic at the start of processing,
even though many other queueing research papers do so. In this paper,
we assume that the processing time of jobs could be arbitrary, and
could be revealed deterministically upon arrival. Furthermore, the
3D printer that prints customized parts usually receives jobs with
different processing requirements. Job arrivals thus could be time-varying,
non-renewal, in batches, dependent, or even arrivals without a pattern.
It motivates us to consider the generic polling system without imposing
any stochastic assumptions on future arrivals.

In addition to 3D printing, many other examples of such a general
polling system can be found in computer-communication systems, reconfigurable
smart manufacturing systems and smart traffic systems \cite{levy1990polling,boon2011applications,miculescu2019polling,van2012polling,boon2012delays}.
In such systems, job arrivals are arbitrary, processing times of jobs
are revealed deterministically upon arrival, and a setup time occurs
when the server switches from one queue to another. We call such a
polling system ``general'' mainly because we do not impose any stochastic
assumptions on job arrivals or service times. Having a scheduling
policy that works well in such a general setting could prevent the
system from performing erratically when rare events occur. Knowing
the worst-case performance of a policy also aids in designing reliable
systems. There are very few studies discussing the optimal policies
or online scheduling algorithms for the polling system without stochastic
assumptions due to the complexity of analysis \cite{vishnevskii2006mathematical}.
It is still unknown if those scheduling policies designed for specific
polling systems work well in the general setting. In our paper, we
study completion time minimization for the general polling system
from an online optimization perspective and obtain the worst-case
performance (i.e., the competitive ratios) of several widely used
scheduling policies with known long-run average waiting time performance,
such as cyclic exhaustive and gated policies. We also, for the first
time, provide conditions for the existence of constant competitive
ratios for online policies in polling systems. Our work bridges the
scheduling and queueing communities by showing that some queueing
policies that work well under stochastic assumptions also work well
in the general scheduling setting. 

\subsection{Problem Statement}

As mentioned earlier, unprocessed jobs from the same prototype (or
family) can be modeled as a queue. We thus consider a single server
system with $k$ parallel queues, as in Figure \ref{fig:Polling-System}.
The processing time $p_{i}$ of the $i^{th}$ arriving job is revealed
instantly upon its arrival time $r_{i}$. The server can serve the
jobs that are waiting in queues in any order non-preemptively. However,
a setup time $\tau$ is incurred when the server switches from one
queue to another. Information $(r_{j},p_{j})$ about a future job
$j$ (for all $j$) remains unknown to the server until job $j$ arrives
in the system. The objective is to find the service order for jobs
and queues to minimize total completion time over all jobs, where
the completion time of a job is the time period from time 0 to the
time when the job has been served (exits the system). It is assumed
that the total number of jobs is an arbitrary large finite quantity. 

\begin{figure}
\includegraphics[scale=0.5]{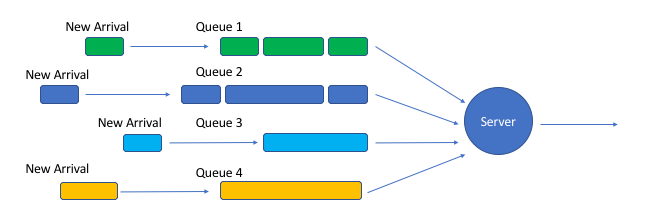}

\caption{Polling System with Four Queues \label{fig:Polling-System}}
\end{figure}

\subsection{Preliminaries}

In this subsection we mainly introduce some concepts and terminologies
that we use in this paper. Important notations in this paper are provided
in Table \ref{tab:Notations}.

\subsubsection*{Machine Scheduling Problems}

Using the notation for one machine scheduling problems \cite{graham1979optimization},
we write our problem as $1\mid r_{i},\tau|\sum C_{i}$, where ``$1$''
refers to the single machine in the system, $r_{i}$ and $\tau$ in
the middle field refer to the release date and setup time constraints,
and $\sum C_{i}$ in the last field is the completion time objective.
In the later part of this paper we will introduce other constraints,
which will appear in the middle field, such as $\tau\leq\theta p_{min}$
and $p_{max}\leq\gamma p_{min}$, where $p_{max}=\sup_{i}p_{i}$ is
the processing time upper bound and $p_{min}=\inf_{i}p_{i}$ is the
processing time lower bound. We say the processing time variation
is bounded if $p_{max}\leq\gamma p_{min}$ for constant $\gamma$.
If no constraint is specified in this field, it means 1) all jobs
are available at time zero, 2) no precedence constraints are imposed,
3) preemption is not allowed, and 4) no setup time exists. In this
paper, we assume jobs and setup times are non-resumable and all the
policies that we discuss are non-preemptive, unless we specifically
state otherwise. Other objectives are $C_{max},$ the maximal completion
time or makespan \cite{lenstra1990approximation} and $\sum w_{i}C_{i},$
the weighted completion time \cite{hall1997scheduling}. 

\subsubsection*{Online and Offline Problems}

A job instance $I$ with $n(I)$ number of jobs is defined as a sequence
of jobs with certain arrival times and processing times, i.e., $I=\{(r_{i},p_{i}),1\leq i\leq n(I)\}$.
In this work we mainly focus on an online scheduling problem, where
$r_{i}$ and $p_{i}$ remain unknown to the server until job $i$
arrives. In contrast to the online problem, the offline problem has
entire information for the job instance $I$, i.e., $(r_{1},r_{2},...,r_{n(I)})$
and $(p_{1},p_{2},...,p_{n(I)})$ from time 0. The offline problem
is usually of great complexity. The offline version of our online
problem, i.e., $1|r_{i},\tau|\sum C_{i}$, is strongly NP-hard since
the offline problem with $\tau=0$ is strongly NP-hard \cite{lawler1993sequencing,hall1997scheduling,kan2012machine}.
However it is important to note that if preemption is allowed, then
\emph{Shortest Remaining Processing Time} (SRPT) is the optimal policy
for $1|r_{i},pmtn|\sum C_{i}$ with $\tau=0$ (see \cite{smith1978new}).
SRPT is online and polynomial-time solvable. It can be used as a benchmark
for online scheduling policies.

\subsubsection*{Scheduling Policies}

A scheduling policy $\pi$ specifies when the server should serve
which job. In our paper we mainly focus on online (or non-anticipative)
policies \cite{wierman2012tail,bansal2018achievable}. Online policies
can be either \emph{deterministic }or \emph{randomized.} For example,
SRPT is a deterministic policy. A randomized policy may toss a coin
before making a decision, and the decision may depend on the outcome
of this coin toss \cite{bansal2018achievable}. Detailed discussions
for randomized algorithms can be found in \cite{stougie2002randomized,chekuri2001approximation}.
In this paper, we focus on deterministic policies.

\begin{table}
\begin{tabular}{|>{\centering}m{2cm}|>{\centering}m{5cm}||>{\centering}m{3cm}|>{\centering}m{5cm}|}
\hline 
Notation & Meaning & Notation & Meaning\tabularnewline
\hline 
\hline 
$r_{i}$ & Release date, the time when job $i$ arrives in the system & $p_{i}$ & Processing time for job $i$\tabularnewline
\hline 
$p_{max}$ & Maximum processing time, $p_{max}=\max_{i}\{p_{i}\}$ & $p_{min}$ & Minimum processing time, $p_{min}=\min_{i}\{p_{i}\}$\tabularnewline
\hline 
$\tau$ & Setup time is $\tau$ for all queues & $I=\{r_{i},p_{i}\}$, for $1\leq i\leq n(I)$ & Job instance, set of jobs with information about release date and
processing time for all jobs in it\tabularnewline
\hline 
$C^{\pi}(I)$ & Total completion time for jobs in instance $I$ under policy $\pi$ & $C^{*}(I)$ & Total completion time for jobs in instance $I$ under the optimal
policy of the offline problem\tabularnewline
\hline 
$C_{i}^{\pi}$ & The completion time for job $i$ under policy $\pi$ & $n(I)$ & Number of jobs in job instance $I$\tabularnewline
\hline 
$\gamma$ & Processing time variation, a constant such that $p_{max}\leq\gamma p_{min}$ & $\theta$ & A constant such that $\tau\leq\theta p_{min}$\tabularnewline
\hline 
$k$ & Number of queues in the system & $\kappa=\max\{\frac{k+1}{k}\gamma,k+1\}$ & Competitive ratio for cyclic-based and exhaustive-like policies, a
constant\tabularnewline
\hline 
\end{tabular}

\caption{Notations \label{tab:Notations}}
\end{table}

\subsubsection*{Competitive Ratios}

A competitive ratio is the ratio between the solution obtained by
an online policy and the \emph{benchmark}. In this paper, the optimal
solution to the offline problem is our benchmark. We say the competitive
ratio for an online policy is $\rho$ if $\sup_{I}\frac{C^{\pi}(I)}{C^{*}(I)}\leq\rho$,
where $C^{\pi}(I)$ is the completion time for job instance $I$ by
a deterministic scheduling policy $\pi$ and $C^{*}(I)$ is the optimal
completion time of the offline problem. We say a competitive ratio
is tight if there exists an instance $I$ such that $\frac{C^{\pi}(I)}{C^{*}(I)}=\rho.$ 

\subsection{Related Work}

Single-machine scheduling problems with setup times or costs have
been widely studied from various perspectives. A detailed review of
the literature can be found in \cite{altman1992elevator,allahverdi2008survey,allahverdi2015third,allahverdi2008significance,ruiz2010hybrid}.
Other studies considering machine setup can be found in \cite{vallada2011genetic,hinder2017novel,ng2003single,mosheiov2005minimizing,shen2018solving}.
However, these papers focus on solving the offline problem where all
the release times and processing times are revealed at time $0$. 

Numerous research papers have shed light on the online single-machine
scheduling problems, without considering setup times. Table \ref{tab:Competitive-Ratios-for-Single-Machine}
summarizes the current state of art of competitive-ratio analysis
over existing online algorithms for these problems. In addition, a
recent paper \cite{lubbecke2016new} provides a new method to approximate
the competitive ratio for general online algorithms. 

\begin{table}
\begin{tabular}{|l|c|c|c|c|}
\hline 
\multirow{2}{*}{Problem} & \multicolumn{2}{c|}{Deterministic} & \multicolumn{2}{c|}{Randomized}\tabularnewline
\cline{2-5} \cline{3-5} \cline{4-5} \cline{5-5} 
 & Lower Bounds & Upper Bounds & Lower Bounds & Upper Bounds\tabularnewline
\hline 
\hline 
$1|r_{i},pmtn|\sum C_{i}$ & 1 & 1 \cite{smith1978new} & 1 & 1 \cite{smith1978new}\tabularnewline
\hline 
$1|r_{i},pmtn|\sum w_{i}C_{i}$ & 1.073 \cite{epstein2003lower} & 1.566 \cite{sitters2010competitive} & 1.038 \cite{epstein2003lower} & $\frac{4}{3}$ \cite{schulz2002power}\tabularnewline
\hline 
$1|r_{i}|\sum C_{i}$ & 2 \cite{hoogeveen1996optimal} & 2 \cite{hoogeveen1996optimal,lu2003class,phillips1998minimizing} & $\frac{e}{e-1}$ \cite{stougie2002randomized} & $\frac{e}{e-1}$ \cite{chekuri2001approximation}\tabularnewline
\hline 
$1|r_{i}|$$\sum w_{i}C_{i}$ & 2 \cite{hoogeveen1996optimal} & 2 \cite{anderson2004online} & $\frac{e}{e-1}$ \cite{stougie2002randomized} & 1.686 \cite{goemans2002single}\tabularnewline
\hline 
$1|r_{i},\frac{p_{max}}{p_{min}}\leq\gamma|\sum w_{i}C_{i}$ & $1+\frac{\sqrt{4\gamma^{2}+1}-1}{2\gamma}$ \cite{tao2010optimal} & $1+\frac{\sqrt{4\gamma^{2}+1}-1}{2\gamma}$ \cite{tao2010optimal} & Unknown & Unknown\tabularnewline
\hline 
$1|r_{i},\frac{p_{max}}{p_{min}}\leq\gamma|\sum C_{i}$ & Numerical \cite{tao2010semi} & $1+\frac{\gamma-1}{1+\sqrt{1+\gamma(\gamma-1)}}$ \cite{tao2010semi} & Unknown & Unknown\tabularnewline
\hline 
\end{tabular}\caption{Competitive Ratios for single-machine Scheduling Problem without Setup
Times (i.e., $\tau=0$)\label{tab:Competitive-Ratios-for-Single-Machine}}
\end{table}

The online makespan minimization problem for the polling system is
considered in \cite{divakaran2011online} and an $\boldsymbol{O}(1)$
policy is proved to exist. However, the competitive ratio provided
in \cite{divakaran2011online} is very large. A 3-competitive online
policy for the polling system with minimizing the completion time
is provided in \cite{zhang2016online}, but only for the case where
$k=2$ and jobs are identical. To the best of our knowledge, online
policies for general polling systems with setup time consideration
have not been well studied.

As we mentioned before, there are articles that study polling systems
from a stochastic perspective by assuming job arrivals, service times
and setup times are stochastic. Long-run average waiting time, queue
length and other metrics of policies are considered for different
types of stochastic assumptions \cite{vishnevskii2006mathematical}.
Exact mean waiting time analysis for cyclic routing policies with
exhaustive, gated and limited service disciplines for polling system
of $M/G/1$ type queues have been provided in \cite{ferguson1985exact,takagi1988queuing,sarkar1989expected,winands2006mean,van2007iterative,gautam2012analysis}.
Service disciplines within queues (such as FCFS, SRPT and others)
are discussed in \cite{wierman2007scheduling}, but routing disciplines
are not discussed. The optimal service policy for the symmetric polling
system is provided by \cite{liu1992optimal}, where ``symmetric''
means that queues and jobs are stochastically identical. However,
the optimal solution for the general polling system remains unknown
\cite{boon2011applications,vishnevskii2006mathematical}. Approximating
algorithms for the polling system are very few, and none of those
widely-used policies have been shown to work well in general settings. 

The contribution of this paper is fourfold: (1) Our work for the first
time analyzes polling systems without stochastic assumptions, evaluates
policy performance by competitive ratios, provides the conditions
for existence of competitive ratios in polling systems, and proves
competitive ratios for some well-studied policies such as cyclic exhaustive
and gated policy. (2) Our work bridges the queueing and scheduling
communities by showing that some widely-used queueing policies also
have decent performance in terms of the worst-case performance in
online scheduling problems. (3) We provide a lower bound for the competitive
ratio for general online policies. (4) We also provide new online
policies that balance future uncertainty and utilize known information,
which may open up a new research direction that would benefit from
revealing information and reducing variability. This paper is organized
as follows: in Section \ref{sec:General-Results} we provide some
general results for scheduling policies in polling systems; in Section
\ref{sec:Polling-System-Large-Setup } we consider policies which
are based on a cyclic routing discipline; in Section \ref{sec:Polling-System-Small}
we consider policies which are processing-time based; we make concluding
remarks in Section \ref{sec:Concluding-Remarks-and}.

\section{\label{sec:General-Results}General Results}

A polling system with zero setup time is the most basic single server
system, in which the server only needs to decide when to serve which
job, without considering the queue indices. From \cite{smith1978new,gittins2011multi}
we know that the optimal online policy is SRPT in this case. The reason
is that by doing this, small jobs are quickly processed, thus the
number of waiting jobs is reduced. However, it becomes problematic
if setup time is non-zero. If a small job arrives at the queue which
requires a setup, then processing this small job before other jobs
may not be beneficial any more. Deciding which queue to set up and
which job to process first thus become challenging in such a polling
system.

In the queueing literature, there are many policies which are defined
by \emph{service disciplines} and \emph{routing discipline}s. The
service discipline of each policy determines the time to switch out
from a queue and the order to serve jobs. The routing discipline determines
the queue to serve next, when the server switches out from a queue.
We first introduce some special routing disciplines which are widely
studied in the literature. 
\begin{itemize}
\item Static disciplines: The server visits each queue (could be empty or
non-empty) following a static routing table, i.e., a predetermined
routing order (see \cite{konheim1994descendant,boxma1993efficient}).
\emph{Cyclic} is one of the most basic static routing disciplines,
under which the server would visit queues one by one, and returns
to the first queue once a cycle has completed. 
\item Random disciplines: The server visits each queue (could be empty or
non-empty) in a random manner (see \cite{konheim1994descendant,chung1994performance,jan2014two}). 
\item Purely queue-length based disciplines: An example is \emph{Stochastic
Largest Queue }(SLQ) policy which is defined by \cite{liu1992optimal},
where the server under SLQ would choose the largest queue to route
to when switching occurs.
\item Purely processing-time based disciplines: The next queue to visit
could be based on the minimal, maximal, total or average processing
time of jobs waiting in the queue. An example is \emph{Shortest Processing
Time} (SPT) policy (see \cite{wierman2007scheduling}), without considering
queue indices.
\end{itemize}
Interestingly, these disciplines have the same competitive ratio lower
bound, of $k$, the number of queues.
\begin{thm}
\label{thm:The-sufficient-conditions}If an online policy has routing
discipline that is: static, random, purely queue-length based, or
purely processing-time based, then this online policy cannot have
a constant competitive ratio that is smaller than $k$ on $1|r_{i},\tau|\sum C_{i}$.
\end{thm}

\begin{proof}
We assume policy $\pi_{1}$ has a routing discipline that is static,
random, or processing time based. We give an instance $I$ where there
is one job at queue $i=1,...,k-1$ at time 0, and $n$ jobs at queue
$k$. Each job has processing time $0$. If $\pi_{1}$ is based on
a static discipline, we assume the routing order is from queue 1 to
queue $k$; if $\pi_{1}$ is based on random routing discipline, we
assume the realization of the random service order is from queue 1
to queue $k;$ if $\pi_{1}$ is purely processing-time based, it treats
all queues equally since the minimal, maximal, total, and average
processing times for all queues are equal, and we again let the server
serve from queue 1 to queue $k$. So the server has to set up $k$
times before reaching queue $k$, and we have

\begin{eqnarray*}
C^{\pi_{1}}(I) & \geq & \tau\left((n+k-1)+(n+k-2)+...+n\right),
\end{eqnarray*}
 and 
\begin{eqnarray*}
C^{*}(I) & = & \tau\left((n+k-1)+(k-1)+...+1\right).
\end{eqnarray*}

Letting $n\rightarrow\infty$, we have $\frac{C^{\pi_{1}}(I)}{C^{*}(I)}\geq k$. 

Now we show the lower bound holds for a policy $\pi_{2}$ that is
purely queue-length based. Suppose at time 0, each queue $i=2,...,k$
has one job with processing time 0, and queue 1 has one job with processing
time $p.$ Suppose the server under $\pi_{2}$ serves from queue 1
to queue $k,$ which is consistent with queue-length based routing.
Then we have 
\begin{eqnarray*}
C^{\pi_{2}}(I) & \geq & (k-1)p+p+\frac{k(k+1)}{2}\tau,
\end{eqnarray*}
and 
\begin{eqnarray*}
C^{*}(I) & = & p+\frac{k(k+1)}{2}\tau.
\end{eqnarray*}

Letting $p\rightarrow\infty$, we have $\frac{C^{\pi_{2}}(I)}{C^{*}(I)}\geq k.$
\end{proof}
It is important to point out that Theorem \ref{thm:The-sufficient-conditions}
provides a competitive ratio bound for policies with some special
routing disciplines, but the service discipline for these policies
could be arbitrary. We will discuss more about policies that are designed
based on service and routing disciplines in Section \ref{sec:Polling-System-Large-Setup }.

To reduce setup frequency, a routing discipline may want to avoid
setting up empty queues, although there may be arrivals during setup
times. Also, an efficient service discipline may prevent the server
from idling when there are unfinished jobs in the system. We thus
consider work-conserving policies, under which the server would never
idle or set up empty queues when the system is non-empty. The next
theorem shows that there exists a constant competitive ratio for all
non-preemptive work-conserving policies, under certain conditions. 
\begin{thm}
\label{thm:Any-non-preemptive-work-conservi}Any non-preemptive work-conserving
policy on the polling system with $p_{max}\leq\gamma p_{min}$ and
$\tau\leq\theta p_{min}$ is at least $(\gamma+\theta)$-competitive
with respect to the optimal solution to the offline problem.
\end{thm}

\begin{proof}
Let $I$ be an arbitrary instance, and let $\hat{I}$ be the processing
time and setup time augmented instance of $I$ such that all job processing
times in $\hat{I}$ are $\hat{p}=p_{max}+\tau\leq(\gamma+\theta)p_{min}$,
setup time is 0 and arrival times in $\hat{I}$ are the same as $I$.
Note that any non-preemptive work-conserving policy on $\hat{I}$
is optimal since processing times for jobs in $\hat{I}$ are identical,
and $\tau=0$. Let $\sigma$ be a non-preemptive work-conserving policy
on $I$, and $\hat{\sigma}$ be a policy that works on $\hat{I}$
and serves jobs in the same order as $\sigma$ does in $I$, without
idling. Then $\hat{\sigma}$ is work-conserving since $\sigma$ never
idles when there are unfinished jobs in system, and therefore $C^{\hat{\sigma}}(\hat{I})=C^{*}(\hat{I})$.
We now show $C^{*}(\hat{I})\leq(\gamma+\theta)C^{*}(I).$ Let $S_{i}^{*}$
and $C_{i}^{*}$ be the starting and completion times of job $i$
in $I$ under the optimal solution. Let $\delta$ be a schedule that
works on $\hat{I}$ and finishes each job $i$ at time $(\gamma+\theta)C_{i}^{*}$,
we then have $C^{\delta}(\hat{I})=(\gamma+\theta)C^{*}(I).$ We now
show that $\delta$ is a feasible schedule on $\hat{I}.$ Notice that
schedule $\delta$ starts serving job $i$ in $\hat{I}$ at time $\hat{S}_{i}^{\delta}=(\gamma+\theta)C_{i}^{*}-\hat{p}=(\gamma+\theta)(S_{i}^{*}+p_{i})-\hat{p}\geq(\gamma+\theta)S_{i}^{*}\geq r_{i}$.
Since $S_{i}^{*}\ge C_{i-1}^{*},$ we also have $(\gamma+\theta)S_{i}^{*}\geq(\gamma+\theta)C_{i-1}^{*}$.
Therefore $\hat{S}_{i}^{\delta}\geq\max\{r_{i},(\gamma+\theta)C_{i-1}^{*}\},$
by induction we can show that $\delta$ is a feasible schedule on
$\hat{I}$. In summary, we have

\begin{eqnarray*}
 &  & C^{\sigma}(I)\leq C^{\hat{\sigma}}(\hat{I})=C^{*}(\hat{I})\leq C^{\delta}(\hat{I})=(\gamma+\theta)C^{*}(I).
\end{eqnarray*}
\end{proof}
Note that Theorem \ref{thm:Any-non-preemptive-work-conservi} holds
only when $p_{max}\leq\gamma p_{min}$ and $\tau\leq\theta p_{min}$.
If either inequality does not hold, we may need other scheduling policies
to achieve constant competitive ratios, which we will discuss in Section
\ref{sec:Polling-System-Large-Setup } and Section \ref{sec:Polling-System-Small}. 

\section{Scheduling Policies with Cyclic Routing \label{sec:Polling-System-Large-Setup }}

In this section we focus on scheduling policies with static routing
disciplines. While these policies are easy to implement, they cannot,
e.g., prioritize small jobs. To characterize the influence made by
job processing time variation, in this section we assume $p_{max}\leq\gamma p_{min}$.
When $\gamma$ is small, all jobs have similar processing times. Unlike
in \cite{tao2010semi} where $p_{min}$ is assumed to be non-zero,
here we also allow $p_{min}=p_{max}=0$, for which we define $\gamma$
to be 1.

In Theorem \ref{thm:The-sufficient-conditions} we have shown that
the competitive ratio for a policy with a static or random routing
discipline is at least $k$, for the system $1|r_{i},\tau|\sum C_{i}$.
Even with additional assumption $p_{max}\leq\gamma p_{min}$, we show
in the following theorem that the competitive ratio for a policy with
a static or random routing discipline is still lower bounded by $k$.
In this case, cyclic routing is the only static discipline that can
achieve this lower bound.
\begin{thm}
\label{thm:Lower-Bound} No online policy with a static or random
routing discipline can guarantee a competitive ratio smaller than
$k$ for $1|r_{i},\tau,p_{max}\leq\gamma p_{min}|\sum C_{i}$, and
cyclic routing  is the only static routing discipline that can achieve
this lower bound. 
\end{thm}

\begin{proof}
Since $1|r_{i},\tau,p_{i}=1|\sum C_{i}$ is a special case of $1|r_{i},\tau,p_{max}\leq\gamma p_{min}|\sum C_{i}$,
we here only need to show the lower bound $k$ holds for $1|r_{i},\tau,p_{i}=1|\sum C_{i}$.
For an arbitrary policy that follows a static routing discipline,
we suppose the server starts from queue 1, and queue $k$ is the last
one visited. Before visiting queue $k$ for the first time, the server
visits queue $i$ for $v_{i}$ times $(1\leq i\leq k-1)$. We construct
a special job instance $I$ by assuming that there is one job arriving
at each queue $i$ every time the server visits queue $i$ for $i=1,...,k-1$.
Also we suppose there are $n_{k}$ jobs at queue $k$ at time 0 for
a large $n_{k}$, and say they form a batch $b_{k}$. If we let $g(b_{k})=\frac{n_{k}(n_{k}+1)}{2},$
then we have

\begin{eqnarray*}
C^{\pi}(I\cup b_{k}) & \geq & C^{\pi}(I)+n_{k}\tau(\sum_{i=1}^{k-1}v_{i}+1)+g(b_{k}),
\end{eqnarray*}

\begin{eqnarray*}
C^{*}(I\cup b_{k}) & \leq & C^{\pi}(I)+n_{k}\tau+g(b_{k})+n(I)(n_{k}+\tau),
\end{eqnarray*}
and $\frac{C^{\pi}(I\cup b_{k})}{C^{*}(I\cup b_{k})}\geq\sum_{i=1}^{k-1}v_{i}+1$
if we let $\tau=(n_{k})^{2}$ and $n_{k}\rightarrow\infty$. Since
for any static or random routing discipline we can construct a job
instance like this, to achieve the smallest ratio $\sum_{i=1}^{k-1}v_{i}+1$,
we need $v_{i}=1$ for $i=1,...,k-1$. Therefore $\frac{C^{\pi}(I\cup b_{k})}{C^{*}(I\cup b_{k})}\geq k$
and only cyclic routing  can achieve this lower bound. 
\end{proof}
Theorem \ref{thm:Lower-Bound} shows the advantage of cyclic routing
 over the other static and random routing disciplines for $1\mid r_{i},\tau,p_{max}\leq\gamma p_{min}\mid\sum C_{i}$.
We now focus our discussion on service disciplines at each queue assuming
cyclic routing. We first discuss a service discipline called \emph{exhaustive},
under which the server serves all the jobs in a queue before switching
out. We show that the exhaustive service has a competitive advantage
over the other disciplines when all job processing times are identical. 
\begin{prop}
\label{lem:Exhaustive is better}For the polling system $1|r_{i},\tau,p_{min}=p_{max}|\sum C_{i}$,
there always exists an exhaustive discipline whose competitive ratio
is at least as small as a non-exhaustive discipline.
\end{prop}

\begin{proof}
The case where $p_{i}=0$ is trivial. Now we let $p_{i}=1$ with appropriate
units. Since preemption is not allowed and all the jobs are identical,
the server only needs to decide when to switch out and which queue
to set up next. If there are jobs in the queue that the server is
currently serving, there are two options for the server: to continue
serving the next job in this queue, or to switch to a queue and later
come back to this queue again. If at time 0 the server is at queue
1 and there is an unfinished job in queue 1, then the server has to
come back after it switches out. Suppose under a non-exhaustive policy
$\pi^{'}$, the server chooses to switch to some queue(s) and come
back to queue 1 at time $T$. Say the server serves instance $I^{'}$
during this period $T$. Suppose there is an adversary policy which
has the same instance at time 0, and this policy chooses to serve
one more job in queue 1, and then follows all decisions that policy
$\pi^{'}$ has made (including waiting). Note that every decision
policy $\pi^{'}$ made is available to the adversary policy because
the adversary policy serves one more job before leaving queue 1. The
total completion time under $\pi^{'}$ is $C^{\pi^{'}}=1+T+C(I^{'})$,
and the completion time achieved by the adversary is $C^{ad}=1+C(I^{'})+n(I^{'})$,
where $C(I^{'})$ is the completion time of instance $I$ served by
policy $\pi^{'}$ during $(0,T]$. We then have $C^{ad}-C^{\pi^{'}}\leq n(I^{'})-T\leq0$.
Notice that the makespan of these two schemes are the same (including
the final setup time of queue 1 for the adversary). Then by induction,
we can show that there always exists an exhaustive discipline that
can achieve a smaller total completion time than a non-exhaustive
one.
\end{proof}
Theorem \ref{thm:Lower-Bound} and Proposition \ref{lem:Exhaustive is better}
motivate us to consider scheduling policies with cyclic routing and
exhaustive service. Define a set of policies $\Pi_{1}$ as those policies
under which the server 1) serves queues in a cyclic way and skips
empty queues when switching, 2) serves each queue exhaustively, 3)
stays in queue $i$ for at most $n_{i}^{w}p_{max}$ amount of time
before switching out, where $n_{i}^{w}$ is the number of jobs processed
in the server's $w^{th}$ visit to queue $i$, and 4) idles at the
last queue that it served if all queues are empty. Note that after
serving all the jobs in queue $i$, the server is allowed to wait
in queue $i$ for some extra time to receive more arrivals. If a new
arrival occurs at queue $i$ during the time that the server is waiting,
then $n_{i}^{w}\leftarrow n_{i}^{w}+1$ and the server can process
this job at any time before it switches out, as long as the server
does not stay in the queue for time longer than the updated $n_{i}^{w}p_{max}$.
If during the $w^{th}$ visit to queue $i$, the processing time for
each job in the queue is $p_{max}$, then the server will switch out
immediately once a queue is exhausted. Also note that we do not specify
the service order of jobs for policies in $\Pi_{1}$. 

Notice that the policies in $\Pi_{1}$ does not set up empty queues,
so the routing discipline for these policies needs to utilize the
queue information. We next consider policies which do not require
any queue information, so the server under these policies would set
up each queue even if the queue is empty before setup is initiated.
Let the set of policies $\Pi_{2}$ be defined as for $\Pi_{1}$ except
that the server sets up regardless of whether the queue is empty or
not. It is important to point out that the policy without waiting
(just exhaustively serving) belongs to $\Pi_{2}$, and the long-run
average waiting time and queue length under this policy have been
extensively studied for $M/G/1$ type queues \cite{ferguson1985exact,takagi1988queuing,sarkar1989expected,winands2006mean}. 

Another widely studied service discipline is called \emph{gated}.
Under a gated discipline, the server only serves the jobs that have
arrived in the queue before the server finishes setting up the queue,
and jobs that arrive after that time will be left to the next cycle
of service. We now let $\Pi_{3}$ be the set of policies with cyclic
routing with skipping the empty queues, and that serves each queue
with a gated discipline. Similar to policy set $\Pi_{1},$ we do not
specify the service order for $\Pi_{3}$ either. Once the server has
set up queue $i$, the number of jobs at that time, $n_{i}^{w}$,
will be served during this visit. We also allow the server to wait
after clearing a queue under $\Pi_{3}$, and the maximal time for
staying in queue $i$ in the $w^{th}$ visit is bounded by $n_{i}^{w}p_{max}$.
Similarly, we can have a policy set $\Pi_{4}$ in which policies are
cyclic and gated, without skipping empty queues when switching. We
do not provide the detailed description of $\Pi_{4}$ here since it
is similar to policy set $\Pi_{3}.$ The long-run average queue length
and waiting time under the gated discipline for $M/G/1$ type queues
are also provided in \cite{takagi1988queuing,winands2006mean}, if
the server does not skip empty queues when switching.

So far we have introduced four policy sets, and we let $\Pi_{r}=\cup_{i=1}^{4}\Pi_{i}.$
In the next theorem we show that all the policies in $\Pi_{r}$ have
the same competitive ratio for problem $1\mid r_{i},\tau,p_{max}\leq\gamma p_{min}\mid\sum C_{i}$.
\begin{thm}
\label{thm:Gated}Any policy $e\in\Pi_{r}$ has the competitive ratio
$\kappa=\max\{\frac{k+1}{k}\gamma,k+1\}$ for the polling system $1\mid r_{i},\tau,p_{max}\leq\gamma p_{min}\mid\sum C_{i}$.
When $\gamma\leq k$, for arbitrary $\epsilon>0$, there is an instance
$I$ such that $\frac{C^{e}(I)}{C^{*}(I)}>\kappa-\epsilon=k+1-\epsilon$.
\end{thm}

\begin{proof}
The detailed proof in given in the Appendix \ref{sec:Proof-for-Theorem}.
\end{proof}
From Theorem \ref{thm:Lower-Bound} we know that the competitive ratio
based on static or random routing is at least $k$, and Theorem \ref{thm:Gated}
shows that any policy from $\Pi_{r}$ has an approximately tight competitive
ratio $k+1$ if $\gamma\leq k$. This indicates that policies from
$\Pi_{r}$ are nearly optimal among all the policies that are based
on static or random routing disciplines. It is also important to point
out that although gated service disciplines are different from exhaustive
disciplines, one can also regard a gated discipline as an exhaustive-like
discipline, as the server under a gated discipline would exhaust the
jobs that have arrived in the previous cycle. The policies in $\Pi_{r}$
have a constant competitive ratio for $1|r_{i},\tau,p_{max}\leq\gamma p_{min}|\sum C_{i}$
because of this exhaustive-like service discipline, since the server
under a policy from $\Pi_{r}$ would serve as many jobs as possible
before switching out. Some other cyclic policies without using an
exhaustive-like service may not have constant competitive ratios for
$1|r_{i},\tau,p_{max}\leq\gamma p_{min}|\sum C_{i}$. We now consider
a policy called the \emph{$l$-limited} policy. This policy is also
based on the cyclic routing. However, the server under the $l$-limited
policy only serves at most $l$ jobs during each visit to a queue.
A detailed description for this policy and the long-run average waiting
time under this policy for $M/G/1$ type queues can be found in \cite{takagi1988queuing,gautam2012analysis,van2007iterative}.
Interestingly, as we shall show in Proposition \ref{thm:The-l-limit},
no constant competitive ratio is guaranteed by the $l$-limited policy,
regardless whether the server skips empty queues or not.
\begin{prop}
\label{thm:The-l-limit}The $l$-limited policy ($l<\infty$, with
or without skipping empty queues) does not have a constant competitive
ratio for $1\mid r_{i},\tau,p_{max}\leq\gamma p_{min}\mid\sum C_{i}$.
\end{prop}

\begin{proof}
We prove this result by giving a special instance $I$. Suppose there
are $(l*n)$ number of jobs $(l,n\in\mathbf{Z_{+})}$ at every queue
at time $0$, and each job has processing time $p=1$. At each queue
the server sets up and serves $l$ jobs, and this is repeated $n$
times. Let $C^{l}(I)$ be the total completion time for the \emph{l-}limited
policy (either with or without skipping empty queues), we have 
\begin{eqnarray*}
\frac{C^{l}(I)}{C^{*}(I)} & = & \frac{\frac{knl(knl+1)}{2}+\tau l\frac{kn(kn+1)}{2}}{\frac{knl(knl+1)}{2}+\tau nl\frac{k(k+1)}{2}}.
\end{eqnarray*}
If we let $\tau=(n)^{2}$ and $n\rightarrow\infty$, then $\frac{C^{l}(I)}{C^{*}(I)}\rightarrow\infty$.
\end{proof}
Proposition \ref{thm:The-l-limit} shows that the \emph{$l$-}limited\emph{
}policy, which does not belong to $\Pi_{r},$ does not have a constant
competitive ratio\emph{ }for $1\mid r_{i},\tau,p_{max}\leq\gamma p_{min}\mid\sum C_{i}$.
Theorem \ref{thm:Gated} and Proposition \ref{thm:The-l-limit} show
the advantage of exhaustive-like service disciplines. However, policies
in $\Pi_{r}$ also have their limitations. We next show that policies
in $\Pi_{r}$ do not have constant competitive ratios if $p_{min}=0<p_{max}$
(so $\gamma$ is infinity).
\begin{thm}
\label{thm:Round-Polices-Large-Workload}Policies in $\Pi_{r}$ do
not guarantee constant competitive ratios for $1|r_{i},\tau|\sum C_{i}$.
\end{thm}

\begin{proof}
We prove the theorem by giving a special job instance $I$. We assume
$p_{min}=0$ and $p_{max}=p$ so that $\gamma=\infty$. Suppose at
time $0$ each of queue $i=2,...,k$ has one job with processing time
$p$ and queue 1 has no job. At time $\tau+\epsilon$ there are $n$
jobs arriving at queue 1, with each having processing time $0$. For
any policy $\pi\in\Pi_{r}$, the server would either setup queue 1
at time 0 then switch to queue 2 at time $\tau$, or setup queue $2$
at time 0. In either of the cases the server will be back to queue
1 when queue $k$ is served in the first cycle. Then we have

\begin{eqnarray*}
C^{\pi}(I) & \geq & \frac{k(k-1)}{2}p+n(k-1)p+\tau\left((n+k-1)+(n+k-2)+...+n\right),
\end{eqnarray*}

and

\begin{eqnarray*}
C^{*}(I) & = & \frac{k(k-1)}{2}p+\tau\left((n+k-1)+(k-1)+...+1\right)+\epsilon(n+k-1).
\end{eqnarray*}

Letting $p=(n)^{2}$ and $n\rightarrow\infty$, we have $\frac{C^{\pi}(I)}{C^{*}(I)}\rightarrow\infty.$
\end{proof}
Theorem \ref{thm:Round-Polices-Large-Workload} shows the limitation
of $\Pi_{r}$ when the condition $p_{max}\leq\gamma p_{min}$ is not
satisfied. Although the server under $\Pi_{r}$ can serve the jobs
within each queue following SPT (see\cite{wierman2007scheduling})
to achieve a smaller expected waiting time, the competitive ratio
remains $\kappa$. When the condition $p_{max}\leq\gamma p_{min}$
no longer holds for finite $\gamma$, one may need policies that utilize
the job processing time information to achieve a constant competitive
ratio. In the next section we will introduce some processing-time
based policies when $p_{max}\leq\gamma p_{min}$ does not hold.

\section{Scheduling Policies Based on Job Processing Times\label{sec:Polling-System-Small}}

In this section we mainly discuss policies in which service and routing
disciplines are based on job processing times. Service and routing
disciplines for these policies are based on job processing times only.
To better characterize the competitive ratio for these policies, in
this section we assume the setup time $\tau$ is bounded by a ratio
of the minimal processing time, that is $\tau\leq\theta p_{min}.$
If $\tau=p_{min}=0,$ we let $\theta=1.$ This setting in the 3D printing
example corresponds to the scenario where jobs of a different color
need to be printed, and the time to set up a new ink is bounded by
a constant factor of the minimum possible processing time. Using the
standard notation for scheduling problems (see \cite{graham1979optimization}),
we denote this polling problem as $1|r_{i},\tau\leq\theta p_{min}|\sum C_{i}$.
Notice that when the setup time is small, i.e., $\theta$ is small,
switching may not be the major contributor to the completion time.
Thus these processing-time based policies may be efficient when $\theta$
is small.

We first introduce a benchmark for deriving the competitive ratio
of our policies. Usually the competitive ratio $\rho$ is defined
by $\sup_{I}\frac{C^{\pi}(I)}{C^{*}(I)}\leq\rho$ where $C^{*}(I)$
is the completion time for $I$ in the offline optimal solution. The
offline problem is strongly NP-hard. To non-rigorously show the NP-hardness,
we know that if no preemption is allowed, even the easier problem
$1|r_{i}|\sum C_{i}$ (without setup time) is strongly NP-hard \cite{lawler1993sequencing,hall1997scheduling,kan2012machine}.
In this section we use a lower bound of the optimal solution as the
benchmark. To get a lower bound for the optimal solution, we introduce
the idea of setup time reduced instance. If instance $I$ is an arbitrary
job instance, then the setup time reduced instance of $I$, say $\utilde{I}$,
is an instance that has the same jobs (same arrival times and same
processing times) as $I$, but has no setup times. The optimal scheduling
policy for $\utilde{I}$ to minimize total completion times is SRPT
\cite{smith1978new}. This scheduling policy is also online, which
is handy for other online policies to emulate. The completion time
of instance $\utilde{I}$ under SRPT is denoted by $C^{p}(\utilde{I})$.
Since setup time does not exist in $\utilde{I}$, we have $C^{p}(\utilde{I})\leq C^{*}(I)$.
In this section, we only consider non-preemptive policies, but using
SRPT as the benchmark. When preemption is not allowed,\emph{ One Machine
}(OM) policy is known to be the optimal online scheduling policy for
$\utilde{I}$ \cite{phillips1998minimizing}. We now apply OM on instance
$I$ and prove its competitive ratio for polling systems. The description
of the OM policy is provided in Algorithm \ref{alg:One-Machine-Scheduling}.
Note that under OM, a job in $I$ can only be scheduled (started)
once it has been completed by SRPT on $\utilde{I}$. The competitive
ratio of OM is provided in Theorem \ref{thm:One-Machine}.

\begin{algorithm}
1. Simulate SRPT policy on the setup time reduced instance $\utilde{I}$.

2. Schedule the jobs non-preemptively in the order of completion time
of jobs by SRPT on $\utilde{I}$.

\caption{One Machine Scheduling (OM) \label{alg:One-Machine-Scheduling}}
\end{algorithm}

\begin{thm}
\label{thm:One-Machine}OM is a $(2+\theta)$-competitive online algorithm
for the polling system $1|r_{i},\tau\leq\theta p_{min}|\sum C_{i}$.
The competitive ratio is tight when using SRPT on the reduced instance
as the benchmark.
\end{thm}

\begin{proof}
Let the completion time of the $j^{th}$ job under OM scheduling be
$C_{j}^{o}$, and the completion time of the $j^{th}$ job completed
under SRPT be $C_{j}^{p}$. Since job $j$ is also the $j^{th}$ job
that completes service under SRPT, we have $\sum_{i=1}^{j}p_{i}\leq C_{j}^{p}$.
Then we have $C_{j}^{o}\leq C_{j}^{p}+\sum_{i:C_{i}^{p}\leq C_{j}^{p}}p_{i}+j\tau\leq(1+\theta)C_{j}^{p}+\sum_{i=1}^{j}p_{i}\leq(2+\theta)C_{j}^{p}$.
Since $C^{p}(\utilde{I})\leq C^{*}(I),$ we get $\sum C_{i}^{o}\leq(2+\theta)\sum C_{i}^{p}\leq(2+\theta)\sum C_{i}^{*}$.
The competitive ratio is tight when there is only one job in the instance
$I$ which is available at time $0$. Suppose this job has processing
time $1$. Then $C^{p}(I)=1$, and $C^{o}(I)=1+(1+\theta)=2+\theta$. 
\end{proof}
The OM algorithm is intuitive, easy to apply and polynomial-time solvable.
Despite its simplicity, we may find it inefficient since setup times
are ignored. Although each unnecessary switch only brings a small
amount of delay if $\theta$ is small, we may still want to avoid
switching too often. Thus we provide another policy that is based
on OM, under which the server will avoid unnecessary setups. We call
it the \emph{Modified One Machine} (MOM) policy, with description
in Algorithm \ref{alg:Gittins-Index-Policy}. Under MOM, we will 1)
regard the completion time of each job under SRPT on $\utilde{I}$
as the new ``arrival'' time, 2) modify the processing time for job
$i$ as $p_{i}+\tau$ if it is not located in the queue that the server
is serving, and 3) process the job with smallest modified processing
time. After completing a job, the processing time for all the jobs
in the system is modified again, and the same mechanism is repeated.
Under MOM, the server will prefer the jobs from the queue that it
is currently serving, thus avoiding frequent switching. We denote
the completion time of job $i$ in $I$ by MOM as $C_{i}^{g}$. 

\begin{algorithm}
\begin{algorithmic}[1]
\Require{Instance $I$} 
\State{Denote the queue that the server is serving as $queue_{server}$}
\While{$I$ has not been fully processed}
\State{Simulate SRPT on $\utilde{I}$. Regard the departure time of the $i^{th}$ job in SRPT as the $i^{th}$ arrival time in MOM. }
\If{$i^{th}$ arrival is at $queue_{server}$}
\State{$\tilde{p}_i = {p_i}$}
\Else
\State{$\tilde{p}_i = {p_i+\tau}$}
\EndIf
\State{Schedule the job with smallest $\tilde{p}_i$}
\EndWhile
\\
\Return{Total completion time $C^{g}(I)$}
\end{algorithmic}

\caption{Modified One Machine Scheduling (MOM) \label{alg:Gittins-Index-Policy}}
\end{algorithm}

\begin{thm}
\label{thm:GIPP}MOM is a $(2+\theta)$-competitive online algorithm
for $1|r_{i},\tau\leq\theta p_{min}|\sum C_{i}$. The competitive
ratio is tight when using SRPT on the reduced instance as the benchmark.
\end{thm}

\begin{proof}
Note both MOM and OM simulate SRPT on $\utilde{I}$ and schedule job
$i$ only after job $i$ has been processed in SRPT. So we can regard
OM as FCFS for a job instance with arrival times \{$C_{i}^{p}$,$i=1,2,...$\},
while MOM serves the job with the smallest modified processing first
in this instance with the same arrival times. We have $C_{j}^{g}\leq C_{j}^{p}+\sum_{i=1}^{j}\tilde{p_{i}}$,
where $\tilde{p_{i}}$ is the modified processing time of job $i$.
Since MOM schedules the available jobs in the descending order of
$\tilde{p}_{i}$, we have $C_{j}^{g}\leq C_{j}^{p}+\sum_{i=1}^{j}\tilde{p_{i}}\leq C_{j}^{p}+\sum_{i:C_{i}^{p}\leq C_{j}^{p}}(p_{i}+\tau)\leq(2+\theta)C_{j}^{p}$.
We give the same example as in Theorem \ref{thm:One-Machine} to show
the tightness of competitive ratio: Suppose there is only one job
with $p=1$ in instance $I$, available at time $0$. Then $C^{p}(I)=1$,
and $C^{g}(I)=1+(1+\theta)=2+\theta$. 
\end{proof}
Though MOM avoids some switching, OM and MOM have the same competitive
ratio when using SRPT on reduced instance as the benchmark. 

Next we show the lower bound for competitive ratios of the problem
$1|r_{i},\tau=\theta p_{min}|\sum C_{i}$. Notice this is a special
case for the problem $1|r_{i},\tau\leq\theta p_{min}|\sum C_{i}$.
\begin{thm}
\label{thm:Lower-Bound-2}If $\tau=\theta p_{min}$ and $\theta\geq0$,
then there is no online algorithm whose competitive ratio is smaller
than $\theta+1$, using SRPT on the reduced instance as the benchmark.
\end{thm}

\begin{proof}
If there is one job with processing time $p_{min}$ in the system,
we have $\frac{C^{\pi}(I)}{C^{p}(I)}\geq\frac{(1+\theta)p_{min}}{p_{min}}=1+\theta$.
If $\tau=p_{min}=0$ (so $\theta=1$), then the lower bounded ratio
is $\theta+1=2$ as provided in \cite{hoogeveen1996optimal}.
\end{proof}
A natural question is whether this lower bound is the best lower bound
that one can have. The answer remains open. There could be either
an online policy whose competitive ratio is exactly equal to this
lower bound, or a larger lower bound which is closer to the ratio
$(2+\theta)$. 

So far we have shown the existence of constant competitive ratios
under different assumptions for polling systems as summarized in Table
\ref{tab:Competitive-Ratios-for}. We find that when either $p_{max}\leq\gamma p_{min}$
or $\tau\leq\theta p_{min}$ holds, we can have constant competitive
ratios. In fact, in many practical scenarios either the processing
time is bounded or the setup time is bounded or both. Also, when $p_{max}\leq\gamma p_{min}$
and $\tau\leq\theta p_{min}$ both hold, any work-conserving policies,
as well as policies from $\Pi_{r}$, OM, and MOM all have constant
competitive ratios. When processing time variation is relatively small
compared with the setup time, i.e., $\gamma$ is small and $\theta$
is large, then policies from $\Pi_{r}$ can have a smaller competitive
ratio than OM or MOM. On the other hand, when $\theta$ is small while
$\gamma$ or $k$ is large, using OM or MOM may be more efficient
as setup times do not contribute much to the total completion time.
When $1\leq\gamma<2$, a work-conserving policy outperforms OM and
MOM, and it also outperforms policies from $\Pi_{r}$ if $\theta<\frac{1}{k}\gamma$
and $\gamma\geq k$. 

\begin{table}
\begin{center}%
\begin{tabular}{|c|c|}
\hline 
Assumption & Competitive Ratio\tabularnewline
\hline 
\hline 
$p_{max}\leq\gamma p_{min}$, $\tau\leq\theta p_{min}$ & $\min\{2+\theta,\gamma+\theta,\max\{\frac{k+1}{k}\gamma,k+1)\}\}$\tabularnewline
\hline 
$\tau\leq\theta p_{min}$ & $2+\theta$\tabularnewline
\hline 
$p_{max}\leq\gamma p_{min}$ & $\max\{\frac{k+1}{k}\gamma,k+1\}$\tabularnewline
\hline 
Unbounded Processing Time and Setup Time & $\geq2$\tabularnewline
\hline 
\end{tabular}\end{center}\caption{Competitive Ratios for Different Cases\label{tab:Competitive-Ratios-for}}
\end{table}

\section{Concluding Remarks and Future Work \label{sec:Concluding-Remarks-and}}

In this paper we consider scheduling policies in the polling system
without stochastic assumptions. Conditions for the existence of constant
competitive ratios are discussed and competitive ratios for several
well-studied polling system scheduling policies are provided. Specifically,
we show that for $1\mid r_{i},\tau,p_{max}\leq\gamma p_{min}\mid\sum C_{i}$
system, an online policy needs to have a cyclic routing discipline
and exhaustive-like service discipline to achieve a constant competitive
ratio. We provide a policy set $\Pi_{r}$ such that every policy from
$\Pi_{r}$ has a constant competitive ratio $\kappa$ for problem
$1\mid r_{i},\tau,p_{max}\leq\gamma p_{min}\mid\sum C_{i}$. We further
provide processing-time based policies which have constant competitive
ratios in system $1\mid r_{i},\tau\leq\theta p_{min}\mid\sum C_{i}$.
We show that if the routing discipline for an online policy is static,
random, purely queue-length based, or purely processing-time based,
then the competitive ratio of this policy cannot be smaller than $k$.
However, it remains unknown whether there exists an online policy
with a constant competitive ratio for the problem $1\mid r_{i},\tau\mid\sum C_{i}$
without any bounding conditions for processing times and setup times.

\bibliographystyle{acm}
\bibliography{JinPolling_Online}

\appendix

\section{\label{sec:Proof-for-Theorem}Proof for Theorem \ref{thm:Gated}
in the Main Paper}

In this section we mainly provide the proof for Theorem \ref{thm:Gated}
of our paper. We first introduce
a fact that will be useful later in our proof. 
\begin{fact}
\label{For-positive-numbers}For positive numbers $\{a_{i},b_{i}\}_{i=1}^{n}$,
we have $\frac{\sum_{i=1}^{n}a_{i}}{\sum_{i=1}^{n}b_{i}}\leq\max_{i=1}^{n}\{\frac{a_{i}}{b_{i}}\}.$
\end{fact}

\begin{proof}
Without loss of generality, assume $\frac{a_{n}}{b_{n}}=\max_{i=1}^{n}\{\frac{a_{i}}{b_{i}}\},$
then for any $i$ we have $\frac{a_{n}}{b_{n}}\geq\frac{a_{i}}{b_{i}}$,
thus $a_{n}b_{i}\geq a_{i}b_{n}$ holds for all $i=1,...,n$. Since
$\frac{a_{n}}{b_{n}}-\frac{\sum_{i=1}^{n}a_{i}}{\sum_{i=1}^{n}b_{i}}=\frac{a_{n}\sum_{i=1}^{n}b_{i}-b_{n}\sum_{i=1}^{n}a_{i}}{b_{n}(\sum_{i=1}^{n}b_{i})}=\frac{\sum_{i=1}^{n}(a_{n}b_{i}-a_{i}b_{n})}{b_{n}(\sum_{i=1}^{n}b_{i})}\geq0$,
we have $\frac{\sum_{i=1}^{n}a_{i}}{\sum_{i=1}^{n}b_{i}}\leq\frac{a_{n}}{b_{n}}=\max_{i=1}^{n}\{\frac{a_{i}}{b_{i}}\}$.
\end{proof}
In Theorem \ref{thm:Gated} we want to show $\sup_{I}\frac{C^{e}(I)}{C^{*}(I)}\leq\rho$
for some constant $\rho.$ However, by Fact \ref{For-positive-numbers}
we only need to show that this inequality holds for the instance processed
in each busy period of the server under an online policy $e\in\Pi_{1}$.
Here we first introduce the concept of busy periods under the online
policy. If there is at least one job in the system, we say the system
is busy, otherwise it is empty. When the system is empty, the server
is not serving under the online policy. The status of the system under
the online policy can be described as a busy period following an empty
period, and then following by a busy period, and so on. There are
two types of busy periods that we are interested in. A Type I busy
period (denoted as I-B) is the busy period in which the server resumes
work without setting up. This is because the server was idling at
the last queue it served (say queue $i$) after the previous busy
period, and the first arrival in the new busy period also occurs at
queue $i$. A Type II busy period (denote as II-B) is the busy period
in which the server resumes work with a setup, because the new arrival
occurs at a queue different from the queue where the server was idling.
Note that the total job instance may be processed by the online policy
in multiple busy periods, and in the proof we only consider the subset
of this job instance that is processed in a certain busy period, which
we call busy period instance. Also note that the very first busy period
must be a II-B, since the server has to set up a queue before processing
at the very beginning.

Without loss of generality, we say a cycle (round) starts when the
server under the online policy visits queue 1 and ends when it visits
queue 1 the next time. If at some time point a queue, say queue $i$,
is empty and skipped by the server, we still say that queue $i$ has
been visited in this cycle, with setup time 0. We define the job instance
served by the server in its $w^{th}$ visit to queue as $b_{i}^{w}$
(for $i=1,...,k$), and we call each $b_{i}^{w}$ a \emph{batch}.
Batch $b_{i}^{w}$ is a subset of job instance $I$. In each cycle,
there are $k$ batches served by the online policy, and some, but
not all may be empty (if all of them are empty then the system is
empty and the server would idle at queue 1). We let $I^{w}=\cup_{i=1}^{k}b_{i}^{w}$.
For each batch $b_{i}^{w}$ with the number of jobs $n(b_{i}^{w})=n_{i}^{w}$,
$S_{i}^{w}$ is the earliest time when a job from $b_{i}^{w}$ starts
being processed under the online policy, and $R_{i}^{w}$ is the earliest
release date (arrival time) over all jobs from batch \textbf{$b_{i}^{w}$}.
Notice that $R_{i}^{w}\leq S_{i}^{w}$. Each batch $b_{i}^{w}$ may
be processed by the optimal offline policy in a different way from
the online policy. Suppose $E_{i}^{w*}$ is the earliest time when
a job in batch $b_{i}^{w}$ starts service under the optimal offline
policy. Note $E_{i}^{w*}$ may differ from $S_{i}^{w}$. Before time
$S_{i}^{w}$, we know all the batches $(\cup_{j=1}^{w-1}I^{j})\cup(\cup_{l=1}^{i-1}b_{l}^{w})$
have been served by the online policy. However in the optimal offline
policy, only some jobs from these batches have been served, and some
other jobs from other cycles of the online policy may have already
been served. We suppose $q_{i}^{w}$ number of jobs in $(\cup_{j=1}^{w-1}I^{j})\cup(\cup_{l=1}^{i-1}b_{l}^{w})$
have been served by the optimal offline policy before time $E_{i}^{w*}$.
We define $E_{i}^{w}$ as the earliest time when the truncated optimal
offline policy starts to serve batch $b_{i}^{w}$, where the idea
of truncated optimal offline policy will be defined later. We let
$C^{e}(I)$ be the cumulative completion times of all jobs in job
instance $I$ under online policy $e\in\Pi_{1}$, and $C^{*}(I)$
be the cumulative completion times of all jobs in $I$ under the offline
optimal policy. For convenience, we let $g(I)=\frac{n(I)\left(n(I)+1\right)}{2}$
for job instance $I$, which is the sum of arithmetic sequence from
1 to $n(I)$. Also, $g(I)$ can be regarded as the completion time
of $I$ when 1) all the jobs are available at time $0$, 2) each of
them has processing time $1$, and 3) no setup time is considered.
All the notations are summarized in Table \ref{tab:List-of-Notations}
of this document. Before going to the proof of Theorem \ref{thm:Gated},
we first provide an example to show how the total completion time
is characterized.
\begin{example}
\label{exa:1}Suppose there is a job instance $I$ with $n(I)=n_{1}+n_{2}$
jobs which arrive at time $R$. Each of the job has processing time
$1$. Under some policy $\pi$, suppose the server starts to process
the first $n_{1}$ jobs, and idles for time $W$, and then processes
the rest of $n_{2}$ jobs without idling. The total completion time
for $I$ is given by
\end{example}

\begin{eqnarray*}
C^{\pi}(I) & = & (R+1)+(R+2)+...+(R+n_{1})+(R+n_{1}+W+1)+(R+n_{1}+W+2)+...+(R+n_{1}+W+n_{2})\\
 & = & (n_{1}+n_{2})R+n_{2}W+g(I).
\end{eqnarray*}

Notice the completion time of $I$ is made up of three components.
The first component $(n_{1}+n_{2})R$ is because all the jobs in $I$
arrive at time $R$. The second term $n_{2}W$ is because the remaining
$n_{2}$ jobs wait for another $W$ amount of time. The third term
$g(I)$ is the pure completion time if we process jobs one by one
without idling.

Notice that the optimal policy may not always be work-conserving (i.e.,
never idles when there are jobs in the system). The optimal policy
may wait at some queue in order to receive more jobs which will arrive
in the future. The \emph{truncated optimal solution} (of a busy period
instance $I$) is defined by the completion time for the optimal offline
problem minus the additional completion time caused by idling, which
is shown in Figure \ref{fig:Truncated-Optimal-Solution}. There is
a waiting (idling) period $W$ between $b_{1}^{1}$ and $b_{1}^{2}$
in Figure \ref{fig:Truncated-Optimal-Solution}. The truncated optimal
solution is given by $C^{*}(b_{1}^{1}\cup b_{1}^{2}\cup b_{2}^{1}\cup b_{3}^{1})-W(n_{1}^{2}+n_{2}^{1}+n_{3}^{1})$.
We use $C^{t}(I)$ to denote the total completion time for a busy
period instance $I$ under the truncated optimal schedule. Note that
this truncated optimal solution is only for an instance $I$ served
by the online policy during a specific busy period, so the earliest
start time of service of this truncated optimal policy cannot be earlier
than the earliest arrival time of $I$. Also note that the truncated
optimal solution is always a lower bound for the real optimal solution. 

\begin{figure}
\begin{center}\includegraphics[scale=0.4]{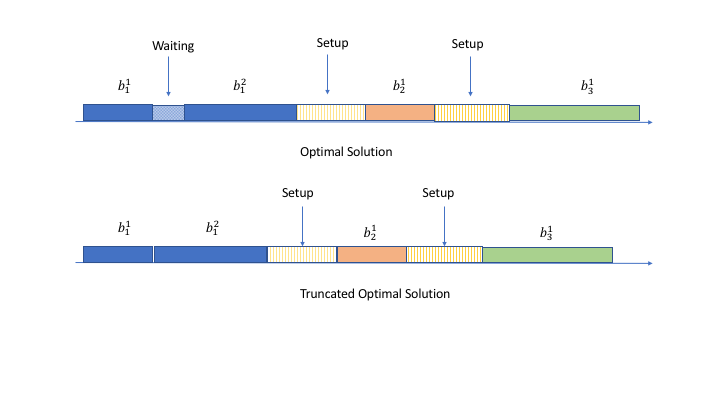}\end{center}

\caption{Truncated Optimal Solution \label{fig:Truncated-Optimal-Solution}}
\end{figure}

\begin{lem}
\label{lem:Truncated}Suppose $I$ is a job instance served in a busy
period by the online policy, $b$ is a batch in some cycle for the
offline problem and $p_{min}=1$, then $C^{t}(I\cup b)\geq C^{t}(I)+g(b)+En(b)+n(b)(n(I)-q),$
where $E$ is the time when the server starts serving batch $b$ in
the truncated optimal solution, and $q$ is the number of jobs in
$I$ that are served before time $E$.
\end{lem}

\begin{proof}
We suppose the optimal solution is given, and now we consider the
total completion time of $I\cup b$ under truncated optimal solution.
If all the jobs from $b$ are served after $I$ in the optimal solution,
then we have $C^{t}(I\cup b)\geq C^{t}(I)+g(b)+En(b)$. If not, we
let $\delta(b)=C^{t}(I\cup b)-C^{t}(I)$ be the additional completion
time incurred by inserting $b$ into $I$. Notice that $\delta(b)$
is minimized when all jobs in $b$ has $p_{min}=1$. If we can show
that $\delta(b)\geq g(b)+En(b)+n(b)(n(I)-q)$ with every job in $b$
having $p_{min}=1$, we can then prove the lemma. So we assume here
that every job in $b$ has $p_{min}=1$. Since the earliest time to
process batch $b$ in the truncated optimal solution is $E,$ if we
combine all jobs in $b$ altogether and serve them in one batch from
time $E$ to time $E+n(b)$, then $\delta(b)$ is again minimized
since all the jobs in $b$ have the smallest processing time. So in
the following we show that by inserting batch $b$ at time $E,$ the
additional completion time incurred is at least $g(b)+En(b)+n(b)(n(I)-q).$
By inserting batch $b$ into $I$ from time $E$ to $E+n(b)$, some
jobs from $I$ served after $E$ in the original truncated optimal
solution (with total number $(n(I)-q)$) are moved after batch $b$,
resulting an increase of delay $n(b)(n(I)-q)$ for these jobs. Besides,
inserting a batch $b$ at time $E$ increases the total completion
time by $g(b)+En(b)$. 
\end{proof}
Lemma \ref{lem:Truncated} is illustrated in Figure \ref{fig:Insert-a-Batch}.
The first schedule in Figure \ref{fig:Insert-a-Batch} is the truncated
optimal for the batch $b_{1}^{1}\cup b_{2}^{1}\cup b_{3}^{1}$. The
second schedule is the truncated optimal schedule for $b_{1}^{1}\cup b_{2}^{1}\cup b_{3}^{1}\cup b$,
where $b$ is separated into two parts. If all jobs in $b$ have processing
time $p_{min}=1$, it is always beneficial to schedule all jobs of
$b$ in the same batch, which is shown as the third schedule in Figure
\ref{fig:Insert-a-Batch}. Notice that in Figure \ref{fig:Insert-a-Batch},
$q=n(b_{1}^{1}+b_{1}^{2})=n_{1}^{1}+n_{1}^{2}.$ 

\begin{figure}
\begin{center}\includegraphics[scale=0.4]{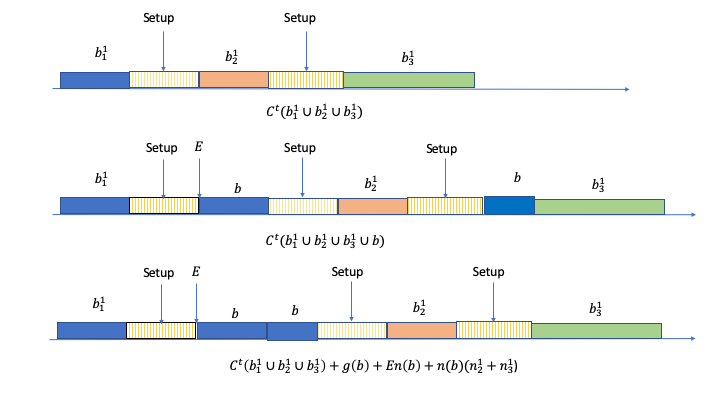}\end{center}

\caption{Insert a Batch\label{fig:Insert-a-Batch}}
\end{figure}

\begin{lem}
\label{lem:Suppose--is}Suppose $I$ is a job instance served by the
online policy in a busy period, $b$ is a batch and $p_{min}=1$,
then $C^{*}(I\cup b)\geq C^{*}(I)+g(b)+E^{*}n(b),$ where $E^{*}$
is the time when the server starts serving batch $b$ in the optimal
solution.
\end{lem}

\begin{proof}
Since the earliest service time in the optimal solution for $b$ is
at $E^{*}$, we have the minimal total completion time for \textbf{$b$
}is $g(b)+E^{*}n(b)$. 
\end{proof}
We now introduce a benchmark for the online policy by combining the
optimal solution and the truncated optimal solution. We let $C^{m}(I)=\alpha C^{*}(I)+(1-\alpha)C^{t}(I)$
be the benchmark, where $\alpha=\frac{1}{1+k}$ (in the proof we will
see why $\alpha$ is given by this constant). We notice that $C^{*}(I)\geq C^{m}(I)\geq C^{t}(I)$
from the fact that $C^{*}(I)\geq C^{t}(I)$. If we have $\frac{C^{e}(I)}{C^{m}(I)}\leq\rho$,
then we can show that $\frac{C^{e}(I)}{C^{*}(I)}\leq\rho$. The reason
for introducing such a benchmark is that one can combine the results
in Lemma \ref{lem:Truncated} and Lemma \ref{lem:Suppose--is} to
provide a lower bound for $C^{m}(I\cup b)$ that contains both $E$
and $E^{*}$, as we will see later in the proof of Theorem \ref{thm:Gated}.

Next we restate Theorem \ref{thm:Gated} in the main paper and describe
the proof. 
\begin{thm}
(\textbf{Theorem \ref{thm:Gated} in the main paper}) Any policy in
$e\in\Pi_{r}$ has the competitive ratio $\kappa=\max\{\frac{k+1}{k}\gamma,k+1\}$
for the polling system $1\mid r_{i},\tau,p_{max}\leq\gamma p_{min}\mid\sum C_{i}$.
When $\gamma\leq k$, for arbitrary $\epsilon>0$, there is an instance
$I$ such that $\frac{C^{e}(I)}{C^{*}(I)}>\kappa-\epsilon=k+1-\epsilon$.
\end{thm}

\begin{proof}
We prove the theorem by induction. In this proof we only consider
an instance $I$ that the online cyclic policy $e\in\Pi_{1}$ serves
in a busy period. By induction we can finally conclude that $\frac{C^{e}(I)}{C^{m}(I)}\leq\kappa$
for instance $I$, which also implies that $\frac{C^{e}(I)}{C^{*}(I)}\leq\kappa.$
We first show that the batches served by the online policy in the
first cycle, i.e., $I^{1}\in I$, satisfies $\frac{C^{e}(I^{1})}{C^{m}(I^{1})}\leq\kappa$.
We next prove that if the result holds for $\cup_{j=1}^{w-1}I^{j}$,
then it also holds for $(\cup_{j=1}^{w-1}I^{j})\cup b_{1}^{w}$. We
then show the result for $(\cup_{j=1}^{w-1}I^{j})\cup(\cup_{i=1}^{l+1}b_{i}^{w})$
if the result holds for $(\cup_{j=1}^{w-1}I^{j})\cup(\cup_{i=1}^{l}b_{i}^{w})$. 

We assume $p_{min}=1$ and $p_{max}=\gamma$ (the case where $p_{min}=0$
is similar). Notice that $I^{1}=\cup_{i=1}^{k}b_{i}^{1}$ is the union
of batches served in the first cycle under the online policy. Without
loss of generality, we assume the server serves from queue 1 to queue
$k$ in each cycle. Knowing that $I^{1}=\cup_{i=1}^{k}b_{i}^{1}$,
we let $n_{(k)}^{1}\geq n_{(k-1)}^{1}\geq...\geq n_{(1)}^{1}$ be
the descending permutation of $(n_{1}^{1},...,n_{k}^{1})$, and $E^{1}=\min_{i=1}^{k}E_{i}^{1}$,
$S^{1}=\min_{i=1}^{k}S_{i}^{1}$. Then we have (with explanation given
later)

\begin{eqnarray}
C^{t}(I^{1}) & \geq & g(I^{1})+\tau\sum_{i=2}^{k}\sum_{j=i}^{k}n_{(k-j+1)}^{1}+E^{1}\sum_{i=1}^{k}n_{i}^{1},\label{eq:1}
\end{eqnarray}
and

\begin{eqnarray}
C^{e}(I^{1}) & \leq & \gamma g(I^{1})+\tau\sum_{i=2}^{k}\sum_{j=i}^{k}n_{(j)}^{1}+S^{1}\sum_{i=1}^{k}n_{i}^{1}.\label{eq:2}
\end{eqnarray}

The RHS of Inequality (\ref{eq:1}) is the minimal completion time
of a list which has the same number of jobs in each queue as $I^{1}$
and all of these jobs arrive at time $E^{1}$ with each job having
processing time 1. The first term $g(I^{1})$ is the pure completion
time. The second term is because $n_{(k)}^{1}\geq n_{(k-1)}^{1}\geq...\geq n_{(1)}^{1}$,
if all batches are available at time $E^{1}$ and there are no further
arrivals, the best order of serving the batches is to serve from the
longest one to the shortest one. Note that the optimal policy may
start without setting up since it may be the same queue that the server
was idling at and resumed with. In the case where there is no setup
for the first queue, the completion time incurred by setup is $\tau\sum_{i=2}^{k}\sum_{j=i}^{k}n_{(k-j+1)}^{1}$.
The third term in the RHS of Inequality (\ref{eq:1}) is because the
entire service process for the truncated optimal solution starts from
$E^{1}$. Therefore, the RHS of Inequality (\ref{eq:1}) is a lower
bound for $C^{t}(I^{1})$. The RHS of inequality (\ref{eq:2}) is
an upper bound for the online policy, which says that the online policy
may serve batches from the shortest to the longest, starting from
time point $S^{1}$, and all the processing times are bounded by  $\gamma$.
Since we consider I-B in this case, the server in the online policy
does not set up for the first batch as the server was idling in the
same queue as the new arrival. So the completion time resulted by
setup is upper bounded by $\tau\sum_{i=2}^{k}\sum_{j=i}^{k}n_{(j)}^{1}$.

We let $Z(k)=\sum_{i=1}^{k}\sum_{j=i}^{k}n_{(j)}^{1}$ and $Z^{t}(k)=\sum_{i=1}^{k}\sum_{j=i}^{k}n_{(k-j+1)}^{1}$
then

\begin{eqnarray*}
\frac{Z(k)}{Z^{t}(k)} & = & \frac{kn_{(k)}^{1}+(k-1)n_{(k-1)}^{1}+...+n_{(1)}^{1}}{kn_{(1)}^{1}+(k-1)n_{(2)}^{1}+...+n_{(k)}^{1}}\leq\frac{kn_{(k)}^{1}+kn_{(k-1)}^{1}+...+kn_{(1)}^{1}}{n_{(1)}^{1}+n_{(2)}^{1}+...+n_{(k)}^{1}}\leq k.
\end{eqnarray*}
From Fact \ref{For-positive-numbers} we have

\begin{eqnarray}
\frac{C^{e}(I^{1})}{C^{t}(I^{1})} & \leq & \frac{\gamma g(I^{1})+\tau\sum_{i=2}^{k}\sum_{j=i}^{k}n_{(i)}^{1}+S^{1}\sum_{i=1}^{k}n_{i}^{1}}{g(I^{1})+\tau\sum_{i=2}^{k}\sum_{j=i}^{k}n_{(k-j+1)}^{1}+E^{1}\sum_{i=1}^{k}n_{i}^{1}}\leq\frac{\gamma g(I^{1})+\tau(Z(k)-\sum_{i=1}^{k}n_{i}^{1})+E^{1}\sum_{i=1}^{k}n_{i}^{1}}{g(I^{1})+\tau(Z^{t}(k)-\sum_{i=1}^{k}n_{i}^{1})+E^{1}\sum_{i=1}^{k}n_{i}^{1}}\nonumber \\
 & \leq & \max\{\gamma,k,1\}<\kappa.\label{eq:4}
\end{eqnarray}

The Inequality (\ref{eq:4}) follows from the fact that $\tau\leq\min_{i=1}^{k}S_{i}^{1}=\min_{i=1}^{k}\{R_{i}^{1}\}\leq\min_{i=1}^{k}E_{i}^{1}$
because there must be a busy period happening before an I-B. 

So far we have shown $\frac{C^{e}(I^{1})}{C^{m}(I^{1})}\leq\kappa$
from the fact that $C^{t}(I^{1})\leq C^{m}(I^{1}).$ Now we prove
that $C^{e}((\cup_{j=1}^{w-1}I^{j})\cup b_{1}^{w})\leq\kappa C^{m}((\cup_{j=1}^{w-1}I^{1})\cup b_{1}^{w})$.
Clearly if $n_{1}^{w}=0$ then the conclusion holds. Now we let $\bar{n}=\sum_{j=1}^{w-1}\sum_{i=1}^{k}n_{i}^{j}+n_{1}^{w}$
and suppose $n_{1}^{w}\neq0$, we then have 

\begin{eqnarray}
C^{e}((\cup_{j=1}^{w-1}I^{j})\cup b_{1}^{w}) & \leq & C^{e}(\cup_{j=1}^{w-1}I^{j})+\gamma g(b_{1}^{w})+n_{1}^{w}\left(\gamma\sum_{i=1}^{k}n_{i}^{w-1}+k\tau+S_{1}^{w-1}\right),\label{eq:5-1}
\end{eqnarray}
and

\begin{eqnarray}
C^{m}((\cup_{j=1}^{w-1}I^{j})\cup b_{1}^{w}) & = & \alpha C^{*}((\cup_{j=1}^{w-1}I^{j})\cup b_{1}^{w})+(1-\alpha)C^{t}((\cup_{j=1}^{w-1}I^{j})\cup b_{1}^{w})\nonumber \\
 & \geq & \alpha C^{*}(\cup_{j=1}^{w-1}I^{j})+\alpha E_{1}^{w*}n_{1}^{w}+\alpha g(b_{1}^{w})\nonumber \\
 &  & +(1-\alpha)C^{t}(\cup_{j=1}^{w-1}I^{j})+(1-\alpha)g(b_{1}^{w})\nonumber \\
 &  & +(1-\alpha)E_{1}^{w}n_{1}^{w}+(1-\alpha)n_{1}^{w}(\bar{n}-q_{1}^{w}),\label{eq:6-1}
\end{eqnarray}

where 
\begin{eqnarray}
E_{1}^{w}\geq & S^{1}+q_{1}^{w},\label{eq:5}
\end{eqnarray}

and

\begin{eqnarray}
E_{1}^{w*}\geq R_{1}^{w}\geq S_{1}^{w-1}+n_{1}^{w-1}.\label{eq:6}
\end{eqnarray}

The RHS of Inequality (\ref{eq:5-1}) is because the completion time
of batch $b_{1}^{w}$ is bounded by $\gamma$ times its pure completion
time $g(b_{1}^{w}),$ which is the maximal pure completion time for
$b_{1}^{w}$ (if processing time for each job in $b_{1}^{w}$ is $p_{max}=\gamma$),
plus $n_{1}^{w}$ times the maximal starting time $\gamma\sum_{i=1}^{k}n_{i}^{w-1}+k\tau+S_{1}^{w-1}$.
Inequality \ref{eq:6-1} follows from Lemma \ref{lem:Truncated} and
\ref{lem:Suppose--is} directly. Inequality (\ref{eq:5}) holds because
before $E_{1}^{w}$, the server has served $q_{1}^{w}$ jobs. Inequality
(\ref{eq:6}) is because the earliest time to serve $b_{1}^{w}$ is
no earlier than $R_{1}^{w}$. From Inequalities (\ref{eq:5-1},\ref{eq:6-1},\ref{eq:5}
and \ref{eq:6}) we have 

\begin{eqnarray*}
\frac{C^{e}((\cup_{j=1}^{w-1}I^{j})\cup b_{1}^{w})}{C^{m}((\cup_{j=1}^{w-1}I^{j})\cup b_{1}^{w})} & \leq & \frac{C^{e}(\cup_{j=1}^{w-1}I^{j})+\gamma g(b_{1}^{w})+n_{1}^{w}\left(\gamma\sum_{i=1}^{k}n_{i}^{w-1}+k\tau+S_{1}^{w-1}\right)}{\alpha C^{*}(\cup_{j=1}^{w-1}I^{j})+(1-\alpha)C^{t}(\cup_{j=1}^{w-1}I^{j})+\alpha S_{1}^{w-1}n_{1}^{w}+g(b_{1}^{w})+(1-\alpha)n_{1}^{w}\bar{n}+(1-\alpha)n_{1}^{w}\tau}\\
 & \leq & \max\{\frac{C^{e}(\cup_{j=1}^{w-1}I^{j})}{C^{m}(\cup_{j=1}^{w-1}I^{j})},\frac{\gamma g(b_{1}^{w})}{g(b_{1}^{w})},\frac{k\tau+S_{1}^{w-1}}{\alpha S_{1}^{w-1}+(1-\alpha)\tau},\frac{\gamma\bar{n}}{(1-\alpha)\bar{n}}\}\\
 & \leq & \max\{\kappa,\gamma,k+1,\frac{\gamma}{1-\alpha}\}.
\end{eqnarray*}

Notice that $\max\{\kappa,\gamma,k+1,\frac{\gamma}{1-\alpha}\}=\max\{\kappa,\gamma,k+1,\gamma\frac{k+1}{k}\}=\kappa$
from the fact $\alpha=\frac{1}{k+1}.$

Now suppose the result holds for $\bar{b}=\left(\cup_{j=1}^{w-1}I^{j}\right)\cup\left(\cup_{i=1}^{l}b_{i}^{w}\right)$
where $w\geq2$, and we want to show it also holds for $\bar{b}\cup b_{l+1}^{w}=\left(\cup_{j=1}^{w-1}I^{j}\right)\cup\left(\cup_{i=1}^{l+1}b_{i}^{w}\right)$
for $l<k$ by induction, where $n_{l+1}^{w}\neq0$. We abuse the notion
by letting $\bar{n}=\sum_{j=1}^{w-1}\sum_{i=1}^{k}n_{i}^{j}+\sum_{i=1}^{l}n_{i}^{w}$
be the number of jobs served before $b_{l+1}^{w}$, and $\bar{n}_{l+1}^{w}=\sum_{j=l+1}^{k}n_{j}^{w-1}+\sum_{j=1}^{l}n_{j}^{w}$
be the number of jobs served between $S_{l+1}^{w-1}$ and $S_{l+1}^{w}$.
Because the server stays in queue $i$ at the $k^{th}$ visit for
no more than time $\gamma n_{i}^{k}$ and serves $n_{i}^{k}$ jobs
,we have

\begin{eqnarray*}
C^{e}(\bar{b}\cup b_{l+1}^{w}) & \leq & C^{e}(\bar{b})+\gamma g(b_{l+1}^{w})+n_{l+1}^{w}\left(\gamma\bar{n}_{l+1}^{w}+k\tau+S_{l+1}^{w-1}\right),
\end{eqnarray*}
and

\begin{eqnarray*}
C^{m}(\bar{b}\cup b_{l+1}^{w}) & = & \alpha C^{*}(\bar{b}\cup b_{l+1}^{w})+(1-\alpha)C^{t}(\bar{b}\cup b_{l+1}^{w})\\
 & \geq & \alpha C^{*}(\bar{b})+\alpha E_{l+1}^{w*}n_{l+1}^{w}+\alpha g(b_{l+1}^{w})\\
 &  & +(1-\alpha)C^{t}(\bar{b})+(1-\alpha)g(b_{l+1}^{w})\\
 &  & +(1-\alpha)E_{l+1}^{w}n_{l+1}^{w}+(1-\alpha)n_{l+1}^{w}(\bar{n}-q_{l+1}^{w}),
\end{eqnarray*}

where 
\begin{eqnarray*}
E_{l+1}^{w} & \geq & \tau+q_{l+1}^{w},
\end{eqnarray*}
and

\begin{eqnarray*}
E_{l+1}^{w*} & \geq R_{l+1}^{w}>S_{l+1}^{w-1}+n_{l+1}^{w-1} & .
\end{eqnarray*}

Similar to our discussion above, we have 

\begin{eqnarray*}
\frac{C^{e}(\bar{b}\cup b_{l+1}^{w})}{C^{m}(\bar{b}\cup b_{l+1}^{w})} & \leq & \frac{C^{e}(\bar{b})+\gamma g(b_{l+1}^{w})+n_{l+1}^{w}\left(\gamma\bar{n}_{l+1}^{w}+k\tau+S_{l+1}^{w-1}\right)}{\alpha C^{*}(\bar{b})+(1-\alpha)C^{t}(\bar{b})+\alpha S_{l+1}^{w-1}n_{l+1}^{w}+g(b_{l+1}^{w})+(1-\alpha)n_{l+1}^{w}\bar{n}+(1-\alpha)n_{1}^{w}\tau}\\
 & \leq & \max\{\frac{C^{e}(\bar{b})}{C^{m}(\bar{b})},\frac{\gamma g(b_{l+1}^{w})}{g(b_{l+1}^{w})},\frac{k\tau+S_{l+1}^{w-1}}{\alpha S_{l+1}^{w-1}+(1-\alpha)\tau},\frac{\gamma\bar{n}}{(1-\alpha)\bar{n}}\}\\
 & \leq & \max\{\kappa,\gamma,k+1,\frac{\gamma}{1-\alpha}\}=\kappa.
\end{eqnarray*}

Now we show that results hold for the second type of busy period,
i.e., II-B. For II-B, the first arrival occurs in a different queue
from where the server was idling, thus the server starts this period
with a setup. Note that the very first busy period is also a II-B.
If it is the first busy period, then we have 
\begin{eqnarray*}
C^{t}(I^{1}) & \geq & g(I^{1})+\tau\sum_{i=1}^{k}\sum_{j=i}^{k}n_{(k-j+1)}^{1},
\end{eqnarray*}
and

\begin{eqnarray*}
C^{e}(I^{1}) & \leq & \gamma g(I^{1})+\tau\sum_{i=1}^{k}\sum_{j=i}^{k}n_{(j)}^{1}+\tau\sum_{i=1}^{k}n_{i}^{1}.
\end{eqnarray*}

Thus 
\begin{eqnarray*}
\frac{C^{e}(I^{1})}{C^{t}(I^{1})} & \leq & \frac{\gamma g(I^{1})+\tau Z(k)+\tau\sum_{i=1}^{k}n_{i}^{1}}{g(I^{1})+\tau Z^{t}(k)}\leq\max\{\gamma,k+1\}\leq\kappa.
\end{eqnarray*}

If it is not the very first busy period, then 
\begin{eqnarray*}
C^{t}(I^{1}) & \geq & g(I^{1})+\tau\sum_{i=2}^{k}\sum_{j=i}^{k}n_{(k-j+1)}^{1}+E^{1}\sum_{i=1}^{k}n_{i}^{1},
\end{eqnarray*}
and

\begin{eqnarray*}
C^{e}(I^{1}) & \leq & \gamma g(I^{1})+\tau\sum_{i=2}^{k}\sum_{j=i}^{k}n_{(j)}^{1}+S^{1}\sum_{i=1}^{k}n_{i}^{1}.
\end{eqnarray*}

If $R^{1}<\tau$, then $S^{1}=R^{1}+\tau\leq2\tau$, and $E^{1}\geq\tau.$
The argument is the same as the very first busy period, which we have
discussed. Now we only need to discuss the case when $R^{1}\geq\tau,$
in which we have 
\begin{eqnarray}
\frac{C^{e}(I^{1})}{C^{t}(I^{1})} & \leq & \frac{\gamma g(I^{1})+\tau\sum_{i=2}^{k}\sum_{j=i}^{k}n_{(j)}^{1}+S^{1}\sum_{i=1}^{k}n_{i}^{1}}{g(I^{1})+\tau\sum_{i=2}^{k}\sum_{j=i}^{k}n_{(k-j+1)}^{1}+E^{1}\sum_{i=1}^{k}n_{i}^{1}}\nonumber \\
 & \leq & \frac{\gamma g(I^{1})+\tau\sum_{i=2}^{k}\sum_{j=i}^{k}n_{(j)}^{1}+(R^{1}+\tau)\sum_{i=1}^{k}n_{i}^{1}}{g(I^{1})+\tau\sum_{i=2}^{k}\sum_{j=i}^{k}n_{(k-j+1)}^{1}+R^{1}\sum_{i=1}^{k}n_{i}^{1}}\label{eq:10}
\end{eqnarray}

From the fact that $\frac{\gamma g(I^{1})+\tau\sum_{i=2}^{k}\sum_{j=i}^{k}n_{(j)}^{1}+(R^{1}+\tau)\sum_{i=1}^{k}n_{i}^{1}}{g(I^{1})+\tau\sum_{i=2}^{k}\sum_{j=i}^{k}n_{(k-j+1)}^{1}+R^{1}\sum_{i=1}^{k}n_{i}^{1}}$
is decreasing as $R^{1}$ increases, we have 
\begin{eqnarray*}
\frac{\gamma g(I^{1})+\tau\sum_{i=2}^{k}\sum_{j=i}^{k}n_{(j)}^{1}+(R^{1}+\tau)\sum_{i=1}^{k}n_{i}^{1}}{g(I^{1})+\tau\sum_{i=2}^{k}\sum_{j=i}^{k}n_{(k-j+1)}^{1}+R^{1}\sum_{i=1}^{k}n_{i}^{1}} & \leq & \frac{\gamma g(I^{1})+\tau\sum_{i=2}^{k}\sum_{j=i}^{k}n_{(j)}^{1}+2\tau\sum_{i=1}^{k}n_{i}^{1}}{g(I^{1})+\tau\sum_{i=2}^{k}\sum_{j=i}^{k}n_{(k-j+1)}^{1}+\tau\sum_{i=1}^{k}n_{i}^{1}}\\
 & \leq & \max\{\gamma,k+1\}\leq\kappa.
\end{eqnarray*}

Discussion for $(\cup_{j=1}^{w-1}I^{j})\cup b_{l+1}^{w}$ for $l=0,...,k-1$
is similar to our discussion for I-B.

Now we prove the approximate tightness argument of this theorem by
constructing a special instance $I$. When $\gamma\leq k$, we have
$\kappa=k+1$. Suppose at time 0 there is one job with processing
time $\gamma$ at queue 1 and one job with $p=1$ at the other queues.
At time $\gamma+\tau+\epsilon_{1}$ (with small $\epsilon_{1}>0$)
a batch $b_{1}^{2}$ arrives at queue 1 and each job in $b_{1}^{2}$
has processing time  $p=1$. We thus have 
\begin{eqnarray*}
C^{e}(I) & \geq & g(I)+n_{1}^{2}(k+1)\tau+\frac{k(k+1)}{2}\tau,
\end{eqnarray*}

and

\begin{eqnarray*}
C^{*}(I) & \leq & \gamma g(I)+n_{1}^{2}(\tau+\epsilon_{1})+\frac{k(k+1)}{2}\tau+(k-1)\epsilon_{1}.
\end{eqnarray*}

Thus when $\tau=(n_{1}^{2})^{2}$ there is an $n_{1}^{2}$ such that
for arbitrary $\epsilon,$
\begin{eqnarray*}
\frac{C^{e}(I)}{C^{*}(I)} & > & 1+k-\epsilon.
\end{eqnarray*}

The theorem also holds for $p_{max}=0$, for simplicity we do not
show the proof here. 

We now provide the proof for $\Pi_{2}$, however this time we only
need to show $C^{ew}(I)\leq\kappa C^{m}(I)$ for any busy period $I$
where $C^{ew}(I)$ is the completion time by a policy from $\Pi_{2}$.
Notice we no longer need to consider different cases for I-B and II-B
because even when the system is idling, the cyclic policy still keeps
setting up queues in cycle. When a new arrival occurs after the system
being empty for some time, we simply regard this time $R^{1}$ as
the beginning of a busy period. Without loss of generality, we assume
the server begins serving with queue 1 in this busy period. Let $n_{(k)}^{1}\geq n_{(k-1)}^{1}\geq...\geq n_{(1)}^{1}$
be the descending permutation of $(n_{1}^{1},...,n_{k}^{1})$ , $R^{1}=\min_{i=1}^{k}R_{i}^{1}$
, $E^{1}=\min_{i=1}^{k}E_{i}^{1}$ and $S^{1}=\min_{i=1}^{k}S_{i}^{1}$.
Notice some of $n_{(i)}^{1}$ may be zero (not all of them) but the
server still sets up the queue even if the queue is empty. We have

\begin{eqnarray*}
C^{m}(I^{1}) & \geq & g(I^{1})+\tau\sum_{i=2}^{k}\sum_{j=i}^{k}n_{(k-j+1)}^{1}+E^{1}\sum_{i=1}^{k}n_{i}^{1},
\end{eqnarray*}
and

\begin{eqnarray*}
C^{ew}(I^{1}) & \leq & \gamma g(I^{1})+\tau k\sum_{i=1}^{k}n_{(i)}^{1}+R^{1}\sum_{i=1}^{k}n_{i}^{1}.
\end{eqnarray*}
To show $\frac{C^{ew}(I^{1})}{C^{m}(I^{1})}\leq\kappa$, by abusing
the notation a little, we first let $Z(k)=k\sum_{i=1}^{k}n_{(i)}^{1}$
and $Z^{t}(k)=\sum_{i=1}^{k}\sum_{j=i}^{k}n_{(k-j+1)}^{1},$ so

\begin{eqnarray*}
 &  & \frac{Z(k)}{Z^{t}(k)}=\frac{kn_{(k)}^{1}+kn_{(k-1)}^{1}+...+kn_{(1)}^{1}}{kn_{(1)}^{1}+(k-1)n_{(2)}^{1}+...+n_{(k)}^{1}}\leq\frac{kn_{(k)}^{1}+kn_{(k-1)}^{1}+...+kn_{(1)}^{1}}{n_{(1)}^{1}+n_{(2)}^{1}+...+n_{(k)}^{1}}\leq k.
\end{eqnarray*}

Thus 
\begin{eqnarray}
\frac{C^{ew}(I^{1})}{C^{m}(I^{1})} & \leq & \frac{\gamma g(I^{1})+\tau k\sum_{i=1}^{k}n_{(i)}^{1}+R^{1}\sum_{i=1}^{k}n_{i}^{1}}{g(I^{1})+\tau\sum_{i=2}^{k}\sum_{j=i}^{k}n_{(k-j+1)}^{1}+E^{1}\sum_{i=1}^{k}n_{i}^{1}}\nonumber \\
 & = & \frac{\gamma g(I^{1})+\tau k\sum_{i=1}^{k}n_{(i)}^{1}+R^{1}\sum_{i=1}^{k}n_{i}^{1}}{g(I^{1})+\tau\sum_{i=2}^{k}\sum_{j=i}^{k}n_{(k-j+1)}^{1}+\max\{\tau,R^{1}\}\sum_{i=1}^{k}n_{i}^{1}}\nonumber \\
 & \leq & \max\{\gamma,k+1\}\label{eq:12}\\
 & \leq & \kappa.\nonumber 
\end{eqnarray}

Inequality (\ref{eq:12}) holds because if $R^{1}\leq\tau$, then
\begin{eqnarray*}
\frac{\gamma g(I^{1})+\tau k\sum_{i=1}^{k}n_{(i)}^{1}+R^{1}\sum_{i=1}^{k}n_{i}^{1}}{g(I^{1})+\tau\sum_{i=2}^{k}\sum_{j=i}^{k}n_{(k-j+1)}^{1}+E^{1}\sum_{i=1}^{k}n_{i}^{1}} & \leq & \frac{\gamma g(I^{1})+\tau k\sum_{i=1}^{k}n_{(i)}^{1}+\tau\sum_{i=1}^{k}n_{i}^{1}}{g(I^{1})+\tau\sum_{i=2}^{k}\sum_{j=i}^{k}n_{(k-j+1)}^{1}+\tau\sum_{i=1}^{k}n_{i}^{1}}\\
 & \leq & \max\{\gamma,k+1\}.
\end{eqnarray*}
And if $R^{1}>\tau,$ then 
\begin{eqnarray*}
\frac{\gamma g(I^{1})+\tau k\sum_{i=1}^{k}n_{(i)}^{1}+R^{1}\sum_{i=1}^{k}n_{i}^{1}}{g(I^{1})+\tau\sum_{i=2}^{k}\sum_{j=i}^{k}n_{(k-j+1)}^{1}+E^{1}\sum_{i=1}^{k}n_{i}^{1}} & \leq & \frac{\gamma g(I^{1})+\tau(k+1)\sum_{i=1}^{k}n_{(i)}^{1}+(R^{1}-\tau)\sum_{i=1}^{k}n_{i}^{1}}{g(I^{1})+\tau\sum_{i=1}^{k}\sum_{j=i}^{k}n_{(k-j+1)}^{1}+(R^{1}-\tau)\sum_{i=1}^{k}n_{i}^{1}}\\
 & \leq & \max\{\gamma,k+1,1\}.
\end{eqnarray*}

The discussions for $(\cup_{j=1}^{w-1}I_{j})\cup b_{l+1}^{w}$ is
similar to proof for $\Pi_{1}$. 

We now show the proof for the policy set $\Pi_{4}$. The proof for
$\Pi_{3}$ is similar to the proof for $\Pi_{1}$. We prove the theorem
by induction. Again for simplicity, we assume $p_{min}=1$ so that
$p_{max}=\gamma$. We want to show $C^{g}(I)\leq\kappa C^{m}(I)$
holds for any instance $I$, where $C^{g}(I)$ is the completion time
for any policy $g\in\Pi_{4}$. Again we assume the server routes from
queue 1 to queue $k$ in each cycle. Let $n_{(k)}^{1}\geq n_{(k-1)}^{1}\geq...\geq n_{(1)}^{1}$
is the descending permutation of $(n_{1}^{1},...,n_{k}^{1})$ and
$E^{1}=\min_{i=1}^{k}E_{i}^{1}$ , $S^{1}=\min_{i=1}^{k}S_{i}^{1}$,
we have

\begin{eqnarray*}
C^{t}(I^{1}) & \geq & g(I^{1})+\tau\sum_{i=2}^{k}\sum_{j=i}^{k}n_{(k-j+1)}^{1}+E^{1}\sum_{i=1}^{k}n_{i}^{1},
\end{eqnarray*}
and

\begin{eqnarray*}
C^{g}(I^{1}) & \leq & \gamma g(I^{1})+\tau k\sum_{i=1}^{k}n_{(i)}^{1}+S^{1}\sum_{i=1}^{k}n_{i}^{1}.
\end{eqnarray*}
The rest of discussions are similar to the proof for $\Pi_{1}$, except
now we have $E_{l+1}^{w*}>R_{l+1}^{w}$ because the policy is gated.

\begin{table}[h]
\begin{tabular}[b]{|>{\centering}m{3cm}|>{\centering}m{4cm}|>{\centering}m{3cm}|>{\centering}m{4cm}|}
\hline 
Notation & Meaning & Notation & Meaning\tabularnewline
\hline 
$b_{i}^{w},$

$i=1,...,k$ & The job instance that are served by the cyclic online policy during
the $w^{th}$ visit ($w^{th}$ cycle) to queue $i$, within a busy
period & $I^{w}$, $w=1,2,...$ & $I^{w}=\cup_{i=1}^{k}b_{i}^{w}$, the union of instances that are
served by the online policy during the $w^{th}$ cycle within a busy
period\tabularnewline
\hline 
$n_{i}^{w}=n(b_{i}^{w}),$ $i=1,...,k$ & The number of jobs in batch $b_{i}^{w}$ & $\alpha=\frac{1}{k+1}$ & A constant\tabularnewline
\hline 
$S_{i}^{w},$

$i=1,...,k$ & The time when the online policy starts to serve batch $b_{i}^{w}$ & $S^{1}=\min_{i=1}^{k}\{S_{i}^{1}\}$ & The earliest staring time for $I^{1}$ by the online policy\tabularnewline
\hline 
$R_{i}^{w}$, $i=1,...,k$ & The earliest release time (arrival time) of batch $b_{i}^{w}$ & $R^{1}=\min_{i=1}^{k}\{R_{i}^{1}\}$ & The earliest release time of $I^{1}$\tabularnewline
\hline 
$E_{i}^{w}$ & The time when the truncated optimal offline policy starts to serve
batch $b_{i}^{w}$ & $E^{1}=\min_{i=1}^{k}\{E_{i}^{1}\}$ & The earliest time when the truncated optimal offline policy starts
to serve $I^{1}$\tabularnewline
\hline 
$E_{i}^{w*}$ & The time when the optimal offline policy starts to serve batch $b_{i}^{w}$ & $q_{i}^{w}$, $i=1,...,k$ & The jobs in $(\cup_{j=1}^{w-1}I^{j})\cup(\cup_{l=1}^{i-1}b_{l}^{w})$
that have been served by the optimal policy before $E_{i}^{w}$\tabularnewline
\hline 
$g(I)=\frac{n(I)(n(I)+1)}{2}$ & Pure completion time for instance $I$ & Busy period & The time period when the system is non-empty\tabularnewline
\hline 
I-B & Type I busy period. The server (under the online policy) starts the
new busy period without setting up & II-B & Type II busy period. The server (under the online policy) starts the
new busy period by setting up\tabularnewline
\hline 
\end{tabular}

\caption{List of Notations for Appendix\label{tab:List-of-Notations}}
\end{table}
\end{proof}

\end{document}